\documentclass[12pt]{article}

\usepackage{amssymb,amsmath,amsthm,amscd,mathrsfs,amsbsy,amsfonts}
\usepackage{pstricks,pst-node}
\usepackage{amsbsy,young,colortbl,cite}

\usepackage{graphicx}

\usepackage{setspace}

\usepackage{ragged2e}

\usepackage{caption}

\usepackage{amsmath}

\usepackage{subfigure}

\usepackage{amssymb,amsthm}
\usepackage[english]{babel}
\newtheorem{proposition}{Proposition}[section]

\newtheorem{definition}[proposition]{Definition}
\newtheorem{Corollary}[proposition]{Corollary}
\newtheorem{theorem}[proposition]{Theorem}

\newtheorem{example}{Example}

\def\({\left(}
\def\){\right)}
\def\[{\begin{eqnarray}}
\def\]{\end{eqnarray}}

\usepackage{amsmath}
\numberwithin{equation}{section}

\begin{document}

\title{Virasoro symmetries of Multi-Component Gelfand-Dickey systems
}

\author{
\  \  Ling An, Chuanzhong Li\footnote{Corresponding author:lichuanzhong@nbu.edu.cn.}\\[4pt]
\small School of Mathematics and Statistics,  Ningbo University, Ningbo, 315211, China\\[4pt]
}

\date{}

\maketitle

\abstract{In this paper, we mainly study the additional symmetry and $\tau$ functions of a multi-component Gelfand-Dickey hierarchy which includes many classical integrable systems, such as the multi-component KdV hierarchy and the multi-component Boussinesq hierarchy. With other kinds of reductions, we can derive a B type multi-component Gelfand-Dickey hierarchy and a C type multi-component Gelfand-Dickey hierarchy. In our research, the additional flows of the additional symmetries can not all survive. By calculating, we find that the generator of the additional symmetry of the C type multi-component Gelfand-Dickey hierarchy is different from that of the B type multi-component Gelfand-Dickey hierarchy, while the forms of their surviving additional flows are the same.
}\\

{\bf Mathematics Subject Classifications (2010)}:  37K05, 37K10, 35Q53.\\
{\bf Keywords:} multi-component Gelfand-Dickey hierarchy, additional symmetry, string equation, $\tau$ function, Virasoro canstraint.  \\
\allowdisplaybreaks
\setlength{\baselineskip}{12pt}
\tableofcontents

\section{Introduction}

\ \ \ \ Gelfand-Dickey hierarchy was introduced by Gelfand and Dickey\cite{1}, which attracted many people's attention, and then became one of the hottest topics in classical integrable systems. The Hamiltonian theory of the Gelfand-Dickey hierarchy was developing gradually in terms of the Lax pairs, and has become a powerful supports for the study of integrable systems. These theories were introduced in detail in \cite{20}. In fact, many domestic and foreign scholars have done  extensive research on the soliton solutions of the Gelfand-Dickey hierarchy, additional symmetry,  $\tau$ function, B$\ddot{a}$cklund transformation and other related properties\cite{2}. In addition, a lot of research  has been done on supersymmetric Gelfand-Dickey hierarchy,  q-deformed Gelfand-Dickey  hierarchy, etc \cite{3,4,5,6,7,8,11}. Generally speaking, little research has been done on the multi-component Gelfand-Dickey (mcGD) hierarchy. In this paper, the additional symmetries and $\tau$ functions of a multi-component Gelfand-Dickey hierarchy, a B type multi-component Gelfand-Dickey (mcBGD) hierarchy and a C type multi-component Gelfand-Dickey (mcCGD) hierarchy will be studied based on the work done by domestic and foreign  scholars on the KP hierarchy and GD hierarchy.

When the theory of additional symmetries first appeared, it ranked only at the edge of integrable systems. With the continuous study of string equations\cite{12,13,14}, Virasoro constraints\cite{15,16} and other theories, it was found that the additional symmetries played an extremely important role in these theories, and the additional symmetries gradually attracted people's  attention\cite{17,18,19}. A direct application of the additional symmetry of the mcGD hierarchy is to derive the string equation which appears in the study of the string theory. The general expression of the string equation is $[P, Q]=1$, where $P$ and $Q$ are differential operators. For string equation, one of its important characteristics is that it has a close relationship with hierarchies of some integrable equations. The relationship between them is that the string equation is invariant under the flows generated by the equation in the hierarchy\cite{20}.
A portion of the additional symmetries can be reduced to Virasoro symmetries. Virasoro symmetries have a wide range of applications,
of which their roles in $\tau$ function are particularly important\cite{9,10}, because $\tau$ function appears as a partition function or as a generating function in modern problems of mathematics and physics. According to the action of Virassoro symmetries on the $\tau$ function, we can also get the explicit solutions of the Virasoro constainted nonlinear equations in the form of matrix integrals. Such as \cite{22,24}, it was the Virasoro costrainted solution in the Toda chain case and in the KdV case. The additional symmetries can also be used to find the eigenfunctions of the linearized problem and to solve the stability problems\cite{21}.

In this paper, the definitions and properties of a mcGD hierarchy, a mcBGD hierarchy and a mcCGD hierarchy are described respectively. The order of the article is as follows: Firstly, we define their Lax equations by the operator $L$, operator $R$ and appropriate constraints, and introduce the wave operator $\phi$ and discuss some properties of wave functions which naturally lead to their Sato equations. Then, by introducing an Orlov-Shulman's operator $M$, the definition of additional symmetries is given, and many practical properties are derived from the additional symmetries. In the calculation, we find that only a small part of the additional flows can survival, so we give the forms of the surviving additional flows. At the same time, a special additional flow is analyzed and calculated, and an important application of additional symmetries in the string theory is obtained. Finally, through the existence theorem of $\tau$ functions, the additional symmetries of $\tau$ functions are discussed, and the equivalent form of the string equation, Virasoro constraint, is obtained. It should be noted that the mcCGD hierarchy can be divided into odd and even forms.

\section{Multi-Component Gelfand-Dickey hierarchies}

\ \ \ \ The Gelfand-Dickey hierarchy is one of the most important topic in the area of classical integrable systems. The definition of the multi-component Gelfand-Dickey hierarchy is based on a $N$-order differential operator $L$ and a operator $R_{\alpha}$ like these\\
\begin{equation}
\ \ \ \ \ \ \ \ \ \ L=A\partial^{N}+u_{2}\partial^{N-2}+u_{3}\partial^{N-3}+\cdots +u_{0},
\end{equation}
among them, $A=diag(a_{1},a_{2},\cdots ,a_{n})$, $a_{i}$ are nonzero constants,
The diagonal elements of $u_{0}$ are all equal to zero, except in the case of $n=1$. And the $u_{i}$ is an arbitrary $n\times n$ matrix;\\
\begin{equation}
R_{\alpha}=\sum\limits_{i=1}^{\infty}R_{i\alpha}\partial^{-i},\ \ \alpha=1,2,\cdots ,n,
\end{equation}
among them, $ R_{0\alpha}= E_{\alpha}$, $E_{\alpha}$ is a matrix having only one nonzero element on the $(\alpha,\alpha)$ place which is equal to $1$, $R_{\alpha}$ satisfies $[L, R_{\alpha}]=0$.\\
It can be found that $R_{j\alpha}$ exist. Their elements are differential polynomials composed of $u_{i}$, and they have the following properties
\begin{equation}
\ \ R_{\alpha}R_{\beta}=\delta_{\alpha\beta}R_{\alpha},\ \ \ \ \ \ \ \ \ \ \sum\limits_{\alpha=1}^{n}R_{\alpha}=I.
\end{equation}
$L$ and $R_{\alpha}$ are collectively called the Lax operators of the mcGD hierarchy.\\

The first definition of the mcGD hierarchy are the Lax equations
\begin{equation*}
\ \ \partial_{n,\alpha}L=[B_{n,\alpha},L],\ \ \ \ \ \ \ \ \ \ \
\partial_{n,\alpha}R_{\beta}=[B_{n,\alpha},R_{\beta}],\ \ n\neq 0\ mod\ N,\
\end{equation*}
\begin{equation}
\sum\limits_{\alpha=1}^{n}\partial_{jN,\alpha}L=0,\ \ \ \ \ \ \ \ \ \ \ \ \ \
\sum\limits_{\alpha=1}^{n}\partial_{jN,\alpha}R_{\beta}=0,\ \ j=1,2,\cdots,\ \ \ \ \ \
\end{equation}
where $\partial_{n,\alpha}=\frac{\partial}{\partial t_{n,\alpha}}$, and $B_{n,\alpha}$ refers to the differential part of the operator $L^{\frac{n}{N}}R_{\alpha}$,
\begin{equation*}
B_{n,\alpha}=(L^{\frac{n}{N}}R_{\alpha})_{+}.
\end{equation*}
The second definition of the mcGD hierarchy is the zero curvature equation
\begin{equation}
\partial_{n,\alpha}B_{m,\alpha}-\partial_{m,\alpha}B_{n,\alpha}+[B_{m,\alpha},B_{n,\alpha}]=0,\ \ n,m\neq 0\ mod\ N,\ \
\end{equation}
the zero curvature equation of the mcGD hierarchy can be derived from its Lax equations.\\
Note: The variables $x$ and $t_{n,\alpha}$ are not independent, but they are also not equivalent. They have the following relationship
\begin{equation*}
\partial=\sum\limits_{\alpha=1}^{n}a_{\alpha}^{-\frac{1}{N}}\partial_{1,\alpha}.
\end{equation*}
When $N=2$, we can derive the multi-component KdV hierarchy; when $N=3$, we can derive the multi-component Boussinesq hierarchy.

The Lax operators of the mcGD hierarchy can also be expressed in dressing form
\begin{equation}
L=\phi A\partial^{N}\phi^{-1},\ \ \ \ \ R_{\alpha}=\phi E_{\alpha}\phi^{-1}, \ \ \ \ \ \
\end{equation}
where the quasi-differential operator $\phi=\phi(A\partial^{N})=\sum\limits_{i=1}^{\infty}\alpha_{i}(A\partial^{N})^{-i}, \alpha_{0}=I$, $\phi$ is called dressing operator or wave operator.
\begin{proposition}
By dressing transformation, the vector field on $A_{u}$ can be converted to $A_{w}$
\begin{equation*}
\partial_{n,\alpha}\phi=-(L^{\frac{n}{N}}R_{\alpha})_{-}\phi,\ \ \ n\neq 0\ mod\ N,
\end{equation*}
\begin{equation}
\sum\limits_{\alpha=1}^{n}\partial_{jN,\alpha}\phi=0,\ \ \ j=1,2,\cdots,\ \ \ \ \ \ \ \ \ \ \
\end{equation}
the above equations are called the Sato equations of the multi-component Gelfand-Dickey hierarchy.
\end{proposition}
The wave function of the mcGD hierarchy is discussed below.\\
Let's first introduce a series
\begin{equation}
\ \ \xi(t,z)=\sum\limits_{i=1}^{\infty}\sum\limits_{\alpha=1}^{n}t_{i,\alpha}E_{\alpha}z^{i}, \ \ \ \ \ \ \ \ \ \ \ \ \ \ \ \ \ \ \ \ \ \ \ \ \ \ \
\end{equation}
$\xi(t,z)$ has the following properties
\begin{equation}\label{3.4}
\begin{split}
\partial_{n,\alpha}e^{\xi(t,z)}=z^{n}E_{\alpha}e^{\xi(t,z)}, \ \ n\neq 0\ mod\ N, \ \ \ \ \ \\
\sum\limits_{\alpha=1}^{n}\partial_{jN,\alpha}e^{\xi(t,z)}=z^{jN}e^{\xi(t,z)},\ \ j=1,2,\cdots,\ \ \ \\                     \partial^{m}e^{\xi(t,z)}=z^{m}A^{-\frac{m}{N}}e^{\xi(t,z)}.\ \ \ \ \ \ \ \ \ \ \ \ \ \ \ \ \ \ \ \ \ \ \ \
\end{split}
\end{equation}
Then the wave function of the mcGD hierarchy can be defined as
\begin{equation*}
W(t,z)=\phi e^{\xi(t,z)}=\omega(t,z)e^{\xi(t,z)}.
\end{equation*}
\begin{Corollary}
The wave function of the multi-component Gelfand-Dickey hierarchy satisfies
\begin{equation}
\begin{split}
L^{\frac{m}{N}}W=z^{m}W,\ \ \ \ \ \ \ \ \partial_{n,\alpha}W=(L^{\frac{n}{N}}R_{\alpha})_{+}W, \ \ n\neq 0\ mod\ N,\ \ \ \ \ \ \ \ \\
\sum\limits_{\alpha=1}^{n}\partial_{jN,\alpha}W=L^{j}W,\ \ j=1,2,\cdots, \ \ \ \ \ \ \ \ \ \ \ \ \ \ \ \ \ \ \ \ \ \ \ \ \ \ \ \ \ \ \ \ \ \ \ \ \ \ \
\end{split}
\end{equation}
\end{Corollary}
\begin{proof}
According to the formula \eqref{3.4}, we can obtain
\begin{equation*}
A^{\frac{m}{N}}\partial^{m}e^{\xi(t,z)}=z^{m}e^{\xi(t,z)},\ \ \ \ \ \ \ \ \ \ \ \ \ \ \ \ \ \ \ \ \ \ \ \ \ \ \ \
\end{equation*}
so we can get
\begin{equation*}
A^{\frac{m}{N}}\partial^{m}=z^{m}, \ \ \ \ \ \ \ \ \ \ \ \ \ \ \ \ \ \ \ \ \ \ \ \ \ \ \ \ \ \ \ \ \ \ \ \ \ \ \ \ \ \
\end{equation*}
therefore
\begin{equation*}
 \ \ \ L^{\frac{m}{N}}W=\phi A^{\frac{m}{N}}\partial^{m}e^{\xi(t,z)}=A^{\frac{m}{N}}\partial^{m}\phi e^{\xi(t,z)}=z^{m}W,
\end{equation*}
\begin{equation*}
\begin{split}
\ \ \ \ \ \ \ \ \ \ \partial_{n,\alpha}W=&\partial_{n,\alpha}(\phi e^{\xi(t,z)})=(\partial_{n,\alpha}\phi)e^{\xi(t,z)}+\phi(\partial_{n,\alpha}e^{\xi(t,z)})\\
=&-(L^{\frac{n}{N}}R_{\alpha})_{-}\phi e^{\xi(t,z)}+\phi z^{n}E_{\alpha}e^{\xi(t,z)}\\
=&-(L^{\frac{n}{N}}R_{\alpha})_{-}W+\phi A^{\frac{n}{N}}\partial^{n}E_{\alpha}\phi^{-1}\phi e^{\xi(t,z)}\\
=&-(L^{\frac{n}{N}}R_{\alpha})_{-}W+L^{\frac{n}{N}}R_{\alpha}W\\
=&(L^{\frac{n}{N}}R_{\alpha})_{+}W.
\end{split}
\end{equation*}
\begin{equation*}
\begin{split}
\ \ \ \ \ \ \ \ \ \ \ \ \ \ \sum\limits_{\alpha=1}^{n}\partial_{jN,\alpha}W=&\sum\limits_{\alpha=1}^{n}\partial_{jN,\alpha}(\phi e^{\xi(t,z)})
=\phi(\sum\limits_{\alpha=1}^{n}\partial_{jN,\alpha}e^{\xi(t,z)})\\ =&\phi(z^{jN}e^{\xi(t,z)})=L^{j}W.
\end{split}
\end{equation*}
\end{proof}
Next, let's consider the adjoint wave function $W^{*}(t,z)$, where the adjoint symbol $``\ast "$ represents a formal adjoint operator, such as $\partial^{*}=-\partial,\ (\partial^{-1})^{*}=-\partial^{-1},\ (AB)^{*}=B^{*}A^{*}$.\\
The adjoint wave function $W^{*}(t,z)$ of the mcGD hierarchy is defined as
\begin{equation}
W^{*}(t,z)=(\phi^{*})^{-1}e^{-\xi(t,z)}.\ \ \ \ \ \ \ \ \ \ \ \ \ \ \ \ \ \ \ \ \ \ \ \
\end{equation}
We have given the Lax operators of the mcGD hierarchy before, and it is easy to prove that the operators $\partial_{n,\alpha}-(L^{\frac{n}{N}}R_{\alpha})_{+}$ and $\sum\limits_{\alpha=1}^{n}\partial_{jN,\alpha}-L^{j}$ can be expressed in a dressing form
\begin{equation*}
\partial_{n,\alpha}-(L^{\frac{n}{N}}R_{\alpha})_{+}=\phi(\partial_{n,\alpha}-A^{\frac{n}{N}}\partial^{n}E_{\alpha})\phi^{-1},\ \  n\neq 0\ mod\ N, \ \ \ \  \ \ \ \
\end{equation*}
\begin{equation}
\sum\limits_{\alpha=1}^{n}\partial_{jN,\alpha}-L^{j}=\phi(\sum\limits_{\alpha=1}^{n}\partial_{jN,\alpha}-A^{j}\partial^{jN})\phi^{-1},\ \ j=1,2,\cdots.\ \ \ \ \ \ \ \ \ \
\end{equation}
Dressing transformation on both sides of $[\partial_{n,\alpha}-A^{\frac{n}{N}}\partial^{n}E_{\alpha},A\partial^{N}]=0$ and $[\sum\limits_{\alpha=1}^{n}\partial_{jN,\alpha}-A^{j}\partial^{jN},A\partial^{N}]=0$, then we can get the Lax equations of the mcGD hierarchy
\begin{equation*}
[\partial_{n,\alpha}-(L^{\frac{n}{N}}R_{\alpha})_{+},L]=0,\ \  n\neq 0\ mod\ N,\ \ \
\end{equation*}
\begin{equation}
[\sum\limits_{\alpha=1}^{n}\partial_{jN,\alpha}-L^{j},L]=0,\ \  j=1,2,\cdots.\ \ \ \ \ \ \ \
\end{equation}

\subsection{Additional symmetry of the mcGD hierarchies}
\ \ \ \ We introduce an operator
\begin{equation}
\ \ \ \ \Gamma=\sum\limits_{k=1}^{\infty}\sum\limits_{\alpha=1}^{n}kt_{k,\alpha}A^{\frac{k-1}{N}}\partial^{k-1}E_{\alpha},\ \ k\neq 0\ mod\ N.
\end{equation}
The operator $\Gamma$ has the following properties
\begin{equation*}
\partial_{n,\alpha}\Gamma=nA^{\frac{n-1}{N}}\partial^{n-1}E_{\alpha},\ \ \
\partial^{k}\Gamma=\Gamma\partial^{k}+kA^{-\frac{1}{N}}\partial^{k-1}, \ \ \ \  n\neq 0\ mod\ N.
\end{equation*}
It is not difficult to verify that the operators $\Gamma$ and $\partial_{n,\alpha}-A^{\frac{n}{N}}\partial^{n}E_{\alpha}$ are commutative, that is to say,
\begin{equation}
[\partial_{n,\alpha}-A^{\frac{n}{N}}\partial^{n}E_{\alpha},\Gamma]=0, \ \ \  n\neq 0\ mod\ N.\ \
\end{equation}
Dressing transformation on both sides of $[\partial_{n,\alpha}-A^{\frac{n}{N}}\partial^{n}E_{\alpha},\Gamma]=0$, then we can obtain
\begin{equation}
\partial_{n,\alpha}M=[(L^{\frac{n}{N}}R_{\alpha})_{+},M], \ \ \  n\neq 0\ mod\ N,\ \
\end{equation}
where $M=\phi\Gamma\phi^{-1}$, we call it an Orlov-Shulman's operator.
\begin{definition}
The solution of the differential equation
\begin{equation}
\frac{\partial\phi}{\partial_{l,m,\alpha}^{*}}=-(M^{m}L^{\frac{l}{N}}R_{\alpha})_{-}\phi, \ \ (\partial_{l,m,\alpha}^{*}=\frac{\partial}{\partial t_{l,m,\alpha}^{*}})\ \ \ \ \ \ \ \ \ \ \ \ \ \ \ \ \ \ \ \ \ \ \
\end{equation}
is called the additional symmetry of the multi-component Gelfand-Dickey hierarchy.
\end{definition}
Next we consider a case related to the restriction of Virasoro. For the special differential operator $M^{m}L^{\frac{l}{N}}R_{\alpha}$, assuming its negative part disappears and let the operator $(M^{m}L^{\frac{l}{N}}R_{\alpha})_{+}$ act on $W$, we can get an equation related to $z$
\begin{equation}
(M^{m}L^{\frac{l}{N}}R_{\alpha})_{+}W=z^{l}E_{\alpha}\partial^{m}_{z}W.
\end{equation}
Note:This system can be rewritten into a linear equation for the isomonodromy problem.
\begin{Corollary}
Combining with the definition of the additional symmetry of the mcGD hierarchy, we can obtain
\begin{equation*}
\partial_{l,m,\alpha}^{*}L=-[(M^{m}L^{\frac{l}{N}}R_{\alpha})_{-},L], \ \ \ \ \ \ \ \ \ \ \ \ \ \ \ \ \ \ \ \
\end{equation*}
\begin{equation}
\partial_{l,m,\alpha}^{*}R_{\beta}=-[(M^{m}L^{\frac{l}{N}}R_{\alpha})_{-},R_{\beta}], \ \ \ \ \ \ \ \ \ \ \ \ \ \ \ \ \
\end{equation}
these imply
\begin{equation}
\ \ \ \ \ \ \partial_{l,m,\alpha}^{*}(L^{\frac{n}{N}}R_{\beta})_{-}=-[(M^{m}L^{\frac{l}{N}}R_{\alpha})_{-},(L^{\frac{n}{N}}R_{\beta})_{-}+\partial_{n,\beta}]_{-},\ \  n\neq 0\ mod\ N.
\end{equation}
\end{Corollary}
\begin{proof}
By calculating, we can get
\begin{equation*}
\begin{split}
&\partial_{l,m,\alpha}^{*}L\\
=&\partial_{l,m,\alpha}^{*}(\phi A\partial^{N}\phi^{-1})\\
=&(\partial_{l,m,\alpha}^{*}\phi)A\partial^{N}\phi^{-1}+\phi A\partial^{N}(\partial_{l,m,\alpha}^{*}\phi^{-1})\\
=&-(M^{m}L^{\frac{l}{N}}R_{\alpha})_{-}L+L(M^{m}L^{\frac{l}{N}}R_{\alpha})_{-}\\
=&-[(M^{m}L^{\frac{l}{N}}R_{\alpha})_{-},L],\ \ \ \ \ \ \ \ \ \ \ \ \ \ \ \ \ \ \ \ \ \ \ \ \ \ \ \ \ \ \ \ \ \ \ \ \ \ \ \ \ \ \ \ \ \ \ \ \ \ \ \ \ \
\end{split}
\end{equation*}
The second equation can be proved by the same principle. Then we can find
\begin{equation*}
\begin{split}
&\partial_{l,m,\alpha}^{*}(L^{\frac{n}{N}}R_{\beta})_{-}\\
=&-[(M^{m}L^{\frac{l}{N}}R_{\alpha})_{-},(L^{\frac{n}{N}}R_{\beta})_{-}]_{-}-[(M^{m}L^{\frac{l}{N}}R_{\alpha})_{-},\partial_{n,\beta}]_{-}\\
=&-[(M^{m}L^{\frac{l}{N}}R_{\alpha})_{-},(L^{\frac{n}{N}}R_{\beta})_{-}+\partial_{n,\beta}]_{-}. \ \ \ \ \ \ \ \ \ \ \ \ \ \ \ \ \ \ \ \ \ \ \ \ \ \ \ \ \ \ \ \
\end{split}
\end{equation*}
\end{proof}
For the additional flows of the mcGD hierarchy, only a few special additional flows can survive.
\begin{theorem}\label{theorem-1}
In the additional flows of the multi-component Gelfand-Dickey hierarchy, only the flows which satisfy the condition $(M^{m-1}L^{\frac{N+l-1}{N}})_{-}=0$ and are shaped like  $\sum\limits_{\alpha=1}^{n}\partial_{l,m,\alpha}^{*}$ can survive.
\end{theorem}
\begin{proof}
For the Lax operator $L$ of the mcGD hierarchy, it has no negative part, that is to say, its negative part is equal to zero.\\
From above, we can know $\partial_{l,m,\alpha}^{*}L=-[(M^{m}L^{\frac{l}{N}}R_{\alpha})_{-},L]$, so let's consider
\begin{equation*}
\partial_{l,m,\alpha}^{*}L_{-}=-[(M^{m}L^{\frac{l}{N}}R_{\alpha})_{-},L]_{-}=-(\phi[\Gamma^{m}A^{\frac{l}{N}}\partial^{l}E_{\alpha},A\partial^{N}]\phi^{-1})_{-},
\end{equation*}
after some deductions, the upper formula can be reduced to
\begin{equation*}
\partial_{l,m,\alpha}^{*}L_{-}=\frac{N}{N+l}\big[L^{\frac{N+l}{N}}R_{\alpha},M^{m}\big]_{-}=mN(M^{m-1}L^{\frac{N+l-1}{N}}R_{\alpha})_{-},\ \ \ \
\end{equation*}
if and only if $(M^{m-1}L^{\frac{N+l-1}{N}})_{-}=0$, we can get
\begin{equation*}
\sum\limits_{\alpha=1}^{n}\partial_{l,m,\alpha}^{*}L_{-}=mN(M^{m-1}L^{\frac{N+l-1}{N}})_{-}=0.
\end{equation*}
\end{proof}
So we just need to thing about the surviving additional flows $\sum\limits_{\alpha=1}^{n}\partial_{l,m,\alpha}^{*}$ which satisfy the condition $(M^{m-1}L^{\frac{N+l-1}{N}})_{-}=0$.\\
\ \ \ \ According to the theorem \ref{theorem-1}, we can find that assuming that a solution of the mcGD hierarchy is defined by the Virasoro condition $\sum\limits_{\alpha=1}^{n}(M^{m-1}L^{\frac{N+l-1}{N}}R_{\alpha})_{-}=0$, if the constraint $(L^{\frac{N+l}{N}})_{-}=0$ is imposed on this solution, it must also satisfy the $W_{1+\infty}$ symmetry $(M^{m})_{-}=0$. That is, the $W_{1+\infty}$ symmetry is compatible both with the Virasoro condition $(M^{m-1}L^{\frac{N+l-1}{N}})_{-}=0$ and the constraint $(L^{\frac{N+l}{N}})_{-}=0$.
\begin{proposition}\label{2.7}
The additional flows $\sum\limits_{\alpha=1}^{n}\partial_{l,m,\alpha}^{*}$ which satisfy the condition $(M^{m-1}L^{\frac{N+l-1}{N}})_{-}=0$ commute with the flows $\partial_{n,\beta}\ (n\neq 0\ mod\ N)\ $ of the multi-component Gelfand-Dickey hierarchy.
\end{proposition}
\begin{proof}
\begin{equation*}
\ \ \ \ \ \ \ \ \big[\sum\limits_{\alpha=1}^{n}\partial_{l,m,\alpha}^{*},\partial_{n,\beta}\big]\phi=\Big(\sum\limits_{\alpha=1}^{n}\big[\partial_{l,m,\alpha}^{*},\partial_{n,\beta}\big]\Big)\phi=
\sum\limits_{\alpha=1}^{n}\Big(\big[\partial_{l,m,\alpha}^{*},\partial_{n,\beta}\big]\phi\Big),
\end{equation*}
\begin{equation*}
\begin{split}
&[\partial_{l,m,\alpha}^{*},\partial_{n,\beta}]\phi\\
=&\partial_{l,m,\alpha}^{*}(\partial_{n,\beta}\phi)-\partial_{n,\beta}(\partial_{l,m,\alpha}^{*}\phi)\\
=&-\partial_{l,m,\alpha}^{*}((L^{\frac{n}{N}}R_{\beta})_{-}\phi)+\partial_{n,\beta}((M^{m}L^{\frac{l}{N}}R_{\alpha})_{-}\phi)\\
=&-(\partial_{l,m,\alpha}^{*}(L^{\frac{n}{N}}R_{\beta})_{-})\phi-(L^{\frac{n}{N}}R_{\beta})_{-}(\partial_{l,m,\alpha}^{*}\phi)\\
&+(\partial_{n,\beta}(M^{m}L^{\frac{l}{N}}R_{\alpha})_{-})\phi+(M^{m}L^{\frac{l}{N}}R_{\alpha})_{-}(\partial_{n,\beta}\phi)\\
=&[(M^{m}L^{\frac{l}{N}}R_{\alpha})_{-},(L^{\frac{n}{N}}R_{\beta})_{-}+\partial_{n,\beta}]_{-}\phi+(L^{\frac{n}{N}}R_{\beta})_{-}(M^{m}L^{\frac{l}{N}}R_{\alpha})_{-}\phi\\
&+[\partial_{n,\beta},(M^{m}L^{\frac{l}{N}}R_{\alpha})_{-}]_{-}\phi-(M^{m}L^{\frac{l}{N}}R_{\alpha})_{-}(L^{\frac{n}{N}}R_{\beta})_{-}\phi\\
=&[(L^{\frac{n}{N}}R_{\beta})_{-},(M^{m}L^{\frac{l}{N}}R_{\alpha})_{-}]\phi-[(L^{\frac{n}{N}}R_{\beta})_{-},(M^{m}L^{\frac{l}{N}}R_{\alpha})_{-}]_{-}\phi\\
=&0,
\end{split}
\end{equation*}
then we can obtain $\big[\sum\limits_{\alpha=1}^{n}\partial_{l,m,\alpha}^{*},\partial_{n,\beta}\big]\phi=0$.\\
From the above calculation we can find they commute on $\phi$, so they will commute on the whole differential algebra generated by coefficients of $\phi$. The proposition is proved.
\end{proof}
A direct application of the additional symmetry of the mcGD hierarchy is to derive the string equation which appears in the study of the string theory. First of all, we found the following relationships between the Lax operators and the Orlov-Shulman's operator of the mcGD hierarchy\\[5pt]
$[L^{\frac{1}{N}},M]=\phi[A^{\frac{1}{N}}\partial,\Gamma]\phi^{-1}=\phi(A^{\frac{1}{N}}\partial\Gamma)\phi^{-1}=\phi(A^{\frac{1}{N}}A^{-{\frac{1}{N}}})\phi^{-1}=I$,\\[5pt]
$[R_{\alpha},M]=\phi[E_{\alpha},\Gamma]\phi^{-1}=O.$\\[5pt]
Then, we can get
\begin{equation}\label{o}
[L^{\frac{n}{N}},M]=\phi[A^{\frac{n}{N}}\partial^{n},\Gamma]\phi^{-1}=\phi(A^\frac{n}{N}nA^{-\frac{1}{N}}\partial^{n-1})\phi^{-1}=nL^{\frac{n-1}{N}}.
\end{equation}
Furthermore, we can obtain
$[L^{\frac{n}{N}},ML^{-\frac{n-1}{N}}R_{\alpha}]=nR_{\alpha}.$\\
Next, we consider a special additional flow $\big(l=-(n-1),\ l-1=0\ mod\ N\big)$
\begin{equation*}
\sum\limits_{\alpha=1}^{n}\partial_{-(n-1),1,\alpha}^{*}L^{\frac{n}{N}}=-\sum\limits_{\alpha=1}^{n}\big[(ML^{-\frac{n-1}{N}}R_{\alpha})_{-},L^{\frac{n}{N}}\big],\ \ n=0\ mod\ N.
\end{equation*}
Combined with the formula \eqref{o}, we can rewrite it to
\begin{equation*}
\sum\limits_{\alpha=1}^{n}\partial_{-(n-1),1,\alpha}^{*}L^{\frac{n}{N}}=\sum\limits_{\alpha=1}^{n}\big[(ML^{-\frac{n-1}{N}}R_{\alpha})_{+},L^{\frac{n}{N}}\big]+nI, \ \ \ \ \ \ \ \ \ \ \ \ \ \ \
\end{equation*}
and when $n=0\ mod\ N$, there is $L^{\frac{n}{N}}=(L^{\frac{n}{N}})_{+},$
therefore
\begin{equation*}
\sum\limits_{\alpha=1}^{n}\big[L^{\frac{n}{N}},\frac{1}{n}(ML^{-\frac{n-1}{N}}R_{\alpha})_{+}\big]=I. \ \ \ \ \ \ \ \ \ \ \ \ \ \ \ \ \ \ \ \ \ \ \ \ \ \ \ \ \ \ \ \ \ \ \ \ \ \ \ \ \ \
\end{equation*}
so
\begin{equation}\label{p}
\Big[L^{\frac{n}{N}},\frac{1}{n}\big(\sum\limits_{\alpha=1}^{n}(ML^{-\frac{n-1}{N}}R_{\alpha})_{+}\big)\Big]=I. \ \ \ \ \ \ \ \ \ \ \ \ \ \ \ \ \ \ \ \ \ \ \ \ \ \ \ \ \ \ \ \ \ \ \ \ \ \ \ \ \ \
\end{equation}
The equation \eqref{p} is called the string equation of the mcGD hierarchy.
From the above deduction process, we can find that the string equation refers to the condition that the operator is independent of the additional variables.\\

\subsection{$\tau$ function and Virasoro constraint of the mcGD hierarchies}
\ \ \ \ The existence theorem of the $\tau$ function is given below.
\begin{theorem}\cite{20}
Suppose the $\tau$ function is a matrix $\tau=(\tau_{\alpha,\beta})$ and $\tau_{\alpha,\alpha}$ is abbreviated as $\tau$, then there exists a function $\tau(\cdots,t_{s,\gamma},\cdots)$, which makes
\begin{equation}
\begin{cases}
\omega_{\alpha,\alpha}(t,z)=\frac{\tau(\cdots,t_{s,\gamma}-\delta_{\gamma,\beta}\frac{1}{sz^{s}},\cdots)}{\tau(\cdots,t_{s,\gamma},\cdots)},\\[5pt]
\omega_{\alpha,\beta}(t,z)=\frac{\tau_{\alpha,\beta}(\cdots,t_{s,\gamma}-\delta_{\gamma,\beta}\frac{1}{sz^{s}},\cdots)}{z\cdot\tau(\cdots,t_{s,\gamma},\cdots)},\ \alpha\neq\beta,\ \ s\neq 0\ mod\ N,
\end{cases}
\end{equation}
only when $\gamma=\beta$, the time variable $t_{s,\gamma}$ will be moved.\\
 Except for the above case $s\neq 0\ mod\ N$, the other cases all satisfy
\begin{equation}
\sum\limits_{\alpha=1}^{n}\partial_{kN,\alpha}\tau=0,\ \ k=1,2,\cdots.
\end{equation}
\end{theorem}
Next, starting from the additional symmetry of the wave operator $\phi$ and combining with the existence theorem of the $\tau$ function, we study the additional symmetry of the $\tau$ function.
To facilitate subsequent calculations, we first split $\sum\limits_{\alpha=1}^{n}\partial_{l,1,\alpha}^{*}\phi\ (l-1=0\ mod\ N)$ appropriately, where
\begin{equation*}
\begin{split}
\sum\limits_{\alpha=1}^{n}\partial_{l,1,\alpha}^{*}\phi=-\sum\limits_{\alpha=1}^{n}(ML^{\frac{l}{N}}R_{\alpha})_{-}\phi=-\sum\limits_{\alpha=1}^{n}(\phi\sum\limits_{k=1}^{\infty}
kt_{k,\alpha}A^{\frac{k+l-1}{N}}\partial^{k+l-1}E_{\alpha}\phi^{-1})_{-}\phi,\\
\ \ \ \ \ \ \ \ \ \ \ \ \ \ \ \ \ \ \ \ \ \ \ \ \ \ \ \ \ \ \ \ \ k\neq 0\ mod\ N,
\end{split}
\end{equation*}
based on the value of $k$, we can divide $(\phi\sum\limits_{k=1}^{\infty}
kt_{k,\alpha}A^{\frac{k+l-1}{N}}\partial^{k+l-1}E_{\alpha}\phi^{-1})_{-}\phi\ (l<0)$ of the above formula into three parts
\begin{equation*}
O=[\phi,t_{1,\alpha}]A^{\frac{l}{N}}\partial^{l}E_{\alpha}+\sum\limits_{k=1}^{-l}
kt_{k,\alpha}\phi A^{\frac{k+l-1}{N}}\partial^{k+l-1}E_{\alpha},
\end{equation*}
\begin{equation*}
P=(\phi(-l+1)t_{-l+1,\alpha}E_{\alpha}\phi^{-1})_{-}\phi=(-l+1)[\phi,t_{-l+1,\alpha}E_{\alpha}],
\end{equation*}
\begin{equation*}
Q=-\sum\limits_{k=-l+2}^{\infty}kt_{k,\alpha}\partial_{k+l-1,\alpha}\phi,
\end{equation*}
from the above splitting results, we can get
\begin{equation}
\begin{split}
\sum\limits_{\alpha=1}^{n}\partial_{l,1,\alpha}^{*}\omega=&-z^{l}\partial_{z}\omega I-\sum\limits_{\alpha=1}^{n}\big(\sum\limits_{k=1}^{-l}kt_{k,\alpha}z^{k+l-1}\omega E_{\alpha}+(-l+1)[\omega,t_{-l+1,\alpha}E_{\alpha}]\ \ \\
&-\sum\limits_{k=-l+2}^{\infty}kt_{k,\alpha}\partial_{k+l-1,\alpha}\omega\big), \ \ \ k\neq 0\ mod\ N,\ \ \ \ \ \ \ \ \ \ \ \
\end{split}
\end{equation}
accordingly, we can obtain
\begin{equation}
\begin{split} \ \ \sum\limits_{\alpha=1}^{n}\partial_{l,1,\alpha}^{*}\omega_{\gamma,\beta}=&\sum\limits_{\alpha=1}^{n}\big(-z^{l}\partial_{z}\omega_{\gamma,\beta}\delta_{\alpha,\beta}-\sum\limits_{k=1}^{-l}kt_{k,\alpha}z^{k+l-1}\omega_{\gamma,\beta}\delta_{\alpha,\beta}-(-l+1)[\omega,t_{-l+1,\alpha}\delta_{\alpha,\beta}]\\
&+\sum\limits_{k=-l+2}^{\infty}kt_{k,\alpha}\partial_{k+l-1,\alpha}\omega_{\gamma,\beta}\big),\ \ \ k\neq 0\ mod\ N.\ \ \ \
\end{split}
\end{equation}
After the above study of relevant knowledge, it is easy to obtain the proposition on the additional symmetry of the $\tau$ function.
\begin{proposition}\label{r}\cite{20}
The action of additional flows $\sum\limits_{\alpha=1}^{n}\partial_{l,1,\alpha}^{*}\ (l-1=0\ mod\ N)$ with $l<0$ on $(\tau_{\gamma,\beta})$ is given by
\begin{equation}
\begin{split}
\sum\limits_{\alpha=1}^{n}\partial_{l,1,\alpha}^{*}\tau_{\gamma,\beta}=\sum\limits_{\alpha=1}^{n}(\sum\limits_{k=-l+1}^{\infty}kt_{k,\alpha}\partial_{k+l-1,\alpha}+\frac{1}{2}
&\sum\limits_{k+s=-l+1}kst_{k,\alpha}t_{s,\alpha})\tau_{\gamma,\beta}+(\sum\limits_{\alpha=1}^{n}c_{\alpha,\beta})\tau_{\gamma,\beta}.
\end{split}
\end{equation}
where $t_{k,\alpha},\ t_{s,\alpha}$ are arguments of the element $\tau_{\gamma,\beta}$.
\end{proposition}
The method of proving the above proposition is similar to that in the literature\cite{20}, and it will not be repeated here.
Because the string equation is the condition that the operator is independent of the additional variables, we can also obtain the condition that the operators are independent of the additional variables by the way of proving the above proposition, as shown below,
\begin{equation}
L_{l}\tau=0,\ \ l-1=0\ mod\ N,\ l<0,\ \ \ \ \ \ \ \ \ \ \ \ \ \ \ \ \ \ \ \ \ \ \ \ \ \ \ \ \ \ \ \ \ \ \ \ \
\end{equation}
where
\begin{equation*}
\begin{split}
L_{l}=\sum\limits_{\alpha=1}^{n}\big(\sum\limits_{k=-l+1}^{\infty}kt_{k,\alpha}&\partial_{k+l-1,\alpha}+\frac{1}{2}\sum\limits_{k+s=-l+1}kst_{k,\alpha}t_{s,\alpha}+c_{\alpha,\beta}\big),
\end{split}
\end{equation*}
it is obviously equivalent to the string equation, We call it the Virasoro constraint of the mcGD hierarchy. And $L_{l}$ satisfies the Virasoro exchange relation
\begin{equation*}
[L_{-m},L_{-n}]=(-m+n)L_{-(m+n)},\ \ m,n=1,2,\cdots.\ \ \ \
\end{equation*}
There are non-autonomous ODEs, related to the solutions of the Virasoro constraints, similar to the Painlev$\acute{e}$ equations appearing in the theory of equations NLS and KdV. Next, we'll give a brief description based on the classical KdV equation.
\begin{example}\cite{23}
The classical KdV equation
\begin{equation}\label{g}
u_{t}+6uu_{x}+u_{xxx}=0
\end{equation}
can be obtained by Lax pairs $L=A\partial^{2}+u_{0}$ and $\mathbb{A}=-4\partial^{3}-3u\partial-3u_{x}$, where $A=diag(1),\ u_{0}=diag(u(x,t))$.
Now we introduce a Virasoro operator
\begin{equation*}
\hat{L}=L_{-1}=\frac{t_{1,1}^{2}}{2}+\sum\limits_{k=2}^{\infty}kt_{k,\alpha}\partial_{k-2,\alpha}.
\end{equation*}
Let $\tau(\mathbf{x})$ be a solution of the mcGD hierarchy, which satisfies $\hat{L}\tau(\mathbf{x})=0,$ add the limit $\frac{\partial\tau}{\partial t_{2i,1}}=0,\ (i=1,2,\cdots)$ to $\tau(\mathbf{x})$ to make it also a solution of the mcKdV hierarchy. For example, $\tau(\mathbf{x})$ satisfies the bilinear differential equation
\begin{equation}
(D_{t_{1,1}}^{4}+D_{t_{1,1}}D_{t_{3,1}})\tau\cdot\tau=0
\end{equation}
of equation\eqref{g}
Then let $\delta(\mathbf{x})=\tau(\mathbf{x})$ under the constraints $t_{1,1}=x,\ t_{5,1}=-\frac{1}{5},\ and\ t_{j,1}=0,\ (j\neq 1,5)$, it can be found that $\delta(\mathbf{x})$ satisfies the bilinear differential equation
\begin{equation}\label{m1}
(D_{x}^{4}-2x)\delta\cdot\delta=0.
\end{equation}
We take $p=-\frac{d^{2}log(\delta)}{dx^{2}}$, then equation\eqref{m1} can be rewritten into the first Painlev$\acute{e}$ equation $P_{\uppercase\expandafter{\romannumeral1}}$
\begin{equation*}
\frac{d^{2}p}{dx^{2}}=6p^{2}+x.
\end{equation*}
\end{example}

\section{B type multi-component Gelfand-Dickey hierarchies}
\ \ \ \ The second part of this paper has given the Lax operators of the mcGD hierarcy
\begin{equation*}
\ \ \ \ \ \ \ \ \ L=A\partial^{N}+u_{2}\partial^{N-2}+u_{3}\partial^{N-3}+\cdots +u_{0},
\end{equation*}
\begin{equation*}
R_{\alpha}=\sum\limits_{i=1}^{\infty}R_{i\alpha}\partial^{-i},\ \ \alpha=1,2,\cdots ,n,
\end{equation*}
if $L$ and $R_{\alpha}$ also satisfy $L^{*}=-\partial L\partial^{-1},\ R_{\alpha}^{*}=\partial R_{\alpha}\partial^{-1}$, then $L$ and $R_{\alpha}$ at this time are called the Lax operators of the mcBGD hierarchy.\\
From the above Lax operators, we can define the Lax equations of the mcBGD hierarchy
\begin{equation*}
\partial_{n,\alpha}L=[B_{n,\alpha},L],\\ \ \ \ \ \ \
\partial_{n,\alpha}R_{\beta}=[B_{n,\alpha},R_{\beta}],\ \ n\neq 0\ mod\ N,\ \ \ \ \ \ \
\end{equation*}
\begin{equation}
\sum\limits_{\alpha=1}^{n}\partial_{jN,\alpha}L=0,\ \ \ \ \ \ \ \ \
\sum\limits_{\alpha=1}^{n}\partial_{jN,\alpha}R_{\beta}=0,\ \ j=1,2,\cdots,\ \ \ \ \ \ \ \ \ \ \ \
\end{equation}
where $L^{*}=-\partial L\partial^{-1},\ R_{\alpha}^{*}=\partial R_{\alpha}\partial^{-1}$, $\partial_{n,\alpha}=\frac{\partial}{\partial t_{n,\alpha}}$ and $B_{n,\alpha}=(L^{\frac{n}{N}}R_{\alpha})_{+}$.
\begin{proposition}
The B type multi-component Gelfand-Dickey hierarchy has only odd flows.
\end{proposition}
\begin{proof}
From the foregoing, we can see that the Lax equations of the mcBGD hierarchy are
\begin{equation*}
\partial_{n,\alpha}L=[B_{n,\alpha},L],\\ \ \ \ \ \ \
\partial_{n,\alpha}R_{\beta}=[B_{n,\alpha},R_{\beta}],\ \ n\neq 0\ mod\ N,\ \ \ \ \ \ \
\end{equation*}
where $L^{*}=-\partial L\partial^{-1},\ R_{\alpha}^{*}=\partial R_{\alpha}\partial^{-1}$, $\partial_{n,\alpha}=\frac{\partial}{\partial t_{n,\alpha}}$ and $B_{n,\alpha}=(L^{\frac{n}{N}}R_{\alpha})_{+}$.\\
Combined with the constraints of its Lax equations, we can get
\begin{equation*}
(L^{\frac{n}{N}})^{*}=(L^{*})^{\frac{n}{N}}=(-1)^{n}\partial L^{\frac{n}{N}}\partial^{-1},
\end{equation*}
\begin{equation*}
\ \ \ \ \ \ \ \ \ \ \ \ (L^{\frac{n}{N}}R_{\alpha})^{*}=R_{\alpha}^{*}(L^{*})^{\frac{n}{N}}=(-1)^{n}\partial R_{\alpha}L^{\frac{n}{N}}\partial^{-1}.
\end{equation*}
Then we simplify $\partial_{n,\alpha}L^{*}$ and $\partial_{n,\alpha}R_{\alpha}^{*}$ in two ways
\begin{equation}\label{a}
\begin{split}
\partial_{n,\alpha}L^{*}=(\partial_{n,\alpha}L)^{*}=[B_{n,\alpha},L]^{*}=-\partial [B_{n,\alpha},L]\partial^{-1}=[\partial B_{n,\alpha}\partial^{-1},L^{*}],\ \ \ \ \ \ \ \ \ \ \ \ \ \ \ \ \ \ \ \ \ \ \ \\
\partial_{n,\alpha}L^{*}=[B_{n,\alpha},L]^{*}=(B_{n,\alpha}L-L B_{n,\alpha})^{*}=L^{*}B_{n,\alpha}^{*}-B_{n,\alpha}^{*}L^{*}=[-B_{n,\alpha}^{*},L^{*}],\ \ \ \ \ \ \ \ \ \ \ \ \ \\[-15pt]
\end{split}
\end{equation}
\begin{equation}\label{b}
\begin{split}
\partial_{n,\alpha}R_{\beta}^{*}=(\partial_{n,\alpha}R_{\beta})^{*}=[B_{n,\alpha},R_{\beta}]^{*}=\partial [B_{n,\alpha},R_{\beta}]\partial^{-1}=[\partial B_{n,\alpha}\partial^{-1},R_{\beta}^{*}],\ \ \ \ \ \ \ \ \ \ \ \ \ \ \\
\partial_{n,\alpha}R_{\beta}^{*}=[B_{n,\alpha},R_{\beta}]^{*}=(B_{n,\alpha}R_{\beta}-R_{\beta} B_{n,\alpha})^{*}=R_{\beta}^{*}B_{n,\alpha}^{*}-B_{n,\alpha}^{*}R_{\beta}^{*}=[-B_{n,\alpha}^{*},R_{\beta}^{*}],
\end{split}
\end{equation}
comparing equation \eqref{a} and \eqref{b}, we can find
\begin{equation}\label{c}
\ \ \ \ \ B_{n,\alpha}^{*}=-\partial B_{n,\alpha}\partial^{-1},\ \ n\neq 0\ mod\ N.
\end{equation}
Then we start with the definition of $B_{n,\alpha}\ (n\neq 0\ mod\ N)$ and solve its adjoint
\begin{equation}\label{d}
B_{n,\alpha}^{*}=((L^{\frac{n}{N}}R_{\alpha})_{+})^{*}=(-1)^{n}\partial (R_{\alpha}L^{\frac{n}{N}})_{+}\partial^{-1}=(-1)^{n}\partial B_{n,\alpha}\partial^{-1},\ \ n\neq 0\ mod\ N.\
\end{equation}
By analyzing equation \eqref{c} and \eqref{d}, we can find that $n$ can only take odd numbers.
\end{proof}
Therefore, a more concise form of the Lax equations of the mcBGD hierarchy can be obtained
\begin{equation*}
\partial_{2n+1,\alpha}L=[B_{2n+1,\alpha},L],\\ \ \ \ \ \partial_{2n+1,\alpha}R_{\beta}=[B_{2n+1,\alpha},R_{\beta}],\ \ 2n+1\neq 0\ mod\ N,
\end{equation*}
\begin{equation}
\sum\limits_{\alpha=1}^{n}\partial_{(2j^{'}+1)N,\alpha}L=0,\ \ \ \ \ \ \ \
\sum\limits_{\alpha=1}^{n}\partial_{(2j^{'}+1)N,\alpha}R_{\beta}=0,\ \ j^{'}=0,1,2,\cdots,\ \ \ \ \
\end{equation}
where $\partial_{2n+1,\alpha}=\frac{\partial}{\partial t_{2n+1,\alpha}}$ and $B_{2n+1,\alpha}=(L^{\frac{2n+1}{N}}R_{\alpha})_{+}$.\\
The Lax operators of the mcBGD hierarchy can also be expressed in dressing form
\begin{equation*}
L=\phi A\partial^{N}\phi^{-1},\ \ \ \ \ \ \ \ \ \ R_{\alpha}=\phi E_{\alpha}\phi^{-1},\ \ \ \ \ \ \ \ \ \ \ \ \
\end{equation*}
where dressing operator $\phi=\phi(A\partial^{N})=\sum\limits_{i=1}^{\infty}\alpha_{i}(A\partial^{N})^{-i},\ \alpha_{0}=I\ and\ \phi^{*}=\partial \phi^{-1}\partial^{-1}$
which is different from the mcGD hierarchy.\\
Subsequently, the Sato equations $\partial_{2n+1,\alpha}\phi=-(B_{2n+1,\alpha})_{-}\phi\ (2n+1\neq 0\ mod\ N)$ and $\sum\limits_{\alpha=1}^{n}\partial_{(2j^{'}+1)N,\alpha}\phi=0\ (j^{'}=0,1,2,\cdots)$ of the mcBGD hierarchy can also be obtained.\\
Then, the wave function $W(t,z)$ and the adjoint wave function $W^{*}(t,z)$ of the mcBGD hierarchy are given
\begin{equation*}
W(t,z)=\phi e^{\xi(t,z)}=\omega(t,z)e^{\xi(t,z)},\ \ \ \ \ \ \ \ \ \ \ \ \ \ \ \ \ \ \ \
\end{equation*}
\begin{equation}
W^{*}(t,z)=(\phi^{*})^{-1}e^{-\xi(t,z)}=(\partial\phi\partial^{-1})e^{-\xi(t,z)},\ \ \ \ \ \
\end{equation}
where
\begin{equation*}
\ \ \ \ \xi(t,z)=\sum\limits_{i=1}^{\infty}\sum\limits_{\alpha=1}^{n}t_{2i-1,\alpha}E_{\alpha}z^{2i-1}. \ \ \ \ \ \ \ \ \ \ \ \ \ \ \ \ \ \ \ \ \ \ \ \ \
\end{equation*}
It should be noted that the Lax equations of the mcBGD hierarchy can also be deduced from
\begin{equation*}
\begin{split}
L^{\frac{2n+1}{N}}W=z^{2n+1}W,\ \ \ \ \ \partial_{2n+1,\alpha}W=B_{2n+1,\alpha}W,\ \ \ 2n+1\neq 0\ mod\ N,\ \ \ \ \ \ \ \ \ \ \ \\
\sum\limits_{\alpha=1}^{n}\partial_{(2j^{'}+1)N,\alpha}W=L^{2j^{'}+1}W,\ \ j^{'}=0,1,2,\cdots. \ \ \ \ \ \ \ \ \ \ \ \ \ \ \ \ \ \ \ \ \ \ \ \ \ \ \ \ \ \ \ \ \ \ \ \ \ \ \ \ \
\end{split}
\end{equation*}

\subsection{Additional symmetry of the mcBGD hierarchies}
\ \ \ \ Firstly, we give the  Orlov-Shulman's operator $M=\phi\Gamma\phi^{-1}$ of the mcBGD hierarchy, among them, $\Gamma=\sum\limits_{i=1}^{\infty}\sum\limits_{\alpha=1}^{n}(2i-1)t_{2i-1,\alpha}A^{\frac{2i-2}{N}}\partial^{2i-2}E_{\alpha},\ \ 2i-1\neq 0\ mod\ N.$\\
Combining the Orlov-Shulman's operator $M$, we can easily get $[M,L^{\frac{2n+1}{N}}]=-(2n+1)L^{\frac{2n}{N}}$.
\begin{definition}
The solution of the differential equation
\begin{equation}
\frac{\partial\phi}{\partial_{l,m,\alpha}^{*}}=-(D_{l,m,\alpha})_{-}\phi,
\end{equation}
where
\begin{equation*}
\ \ \partial_{l,m,\alpha}^{*}=\frac{\partial}{\partial t_{l,m,\alpha}^{*}},\ \ D_{l,m,\alpha}=M^{m}L^{\frac{l}{N}}R_{\alpha}-(-1)^{l}R_{\alpha}L^{\frac{l-1}{N}}M^{m}L^{\frac{1}{N}},
\end{equation*}
is called the additional symmetry of the B type multi-component Gelfand-Dickey hierarchy.
\end{definition}
Similarly, for the differential operator $D_{l,m,\alpha}$, assuming its negative part disappears and let the operator $(D_{l,m,\alpha})_{+}$ act on $W$, we can get an equation related to $z$
\begin{equation}
(D_{l,m,\alpha})_{+}W=(z^{l}-(-1)^{l}z^{2l-1})E_{\alpha}\partial^{m}_{z}W-(-1)^{l}m(l-1)z^{2(l-1)}E_{\alpha}\partial^{m-1}_{z}W.
\end{equation}
Note:This system can also be rewritten into a linear equation for the isomonodromy problem.\\
According to the definition of the additional symmetry of the mcBGD hierarchy, some formulas can be obtained by simple calculation
\begin{equation}
\partial_{l,m,\alpha}^{*}L=-[(D_{l,m,\alpha})_{-},L], \ \ \ \ \ \ \
\partial_{l,m,\alpha}^{*}R_{\beta}=-[(D_{l,m,\alpha})_{-},R_{\beta}].
\end{equation}
For the additional flows of the mcBGD hierarchy, they have a fraction of the flows that can survive.
\begin{theorem}
In the additional flows of the B type multi-component Gelfand-Dickey hierarchy, the flows which satisfy the condition $(M^{m-1}L^{\frac{N+l-1}{N}})_{-}=0$ and are shaped like $\sum\limits_{\alpha=1}^{n}\partial_{l,m,\alpha}^{*}$ or the condition $l=2i\ (i\in\mathbb{Z})$ and are shaped like $\sum\limits_{\alpha=1}^{n}\partial_{l,m,\alpha}^{*}$ can survive.
\end{theorem}
\begin{proof}
For Lax operator $L$ of the mcBGD hierarchy, its negative part is equal to zero.\\
And we know $\partial_{l,m,\alpha}^{*}L=-[(D_{l,m,\alpha})_{-},L]$, so let's consider
\begin{equation*}
\begin{split}
\ \ \ \ \partial_{l,m,\alpha}^{*}L_{-}=&-[(D_{l,m,\alpha})_{-},L]_{-}=-(\phi[\Gamma^{m}A^{\frac{l}{N}}\partial^{l}E_{\alpha},A\partial^{N}]\phi^{-1})_{-}\\
&+(-1)^{l}(\phi[E_{\alpha}A^{\frac{l-1}{N}}\partial^{l-1}\Gamma^{m}A^{\frac{1}{N}}\partial,A\partial^{N}]\phi^{-1})_{-},
\end{split}
\end{equation*}
after some deductions, the upper formula can be reduced to
\begin{equation*}
\partial_{l,m,\alpha}^{*}L_{-}=(1-(-1)^{l})mN(M^{m-1}L^{\frac{N+l-1}{N}}R_{\alpha})_{-},\ \ \ \ \ \ \ \
\end{equation*}
only when $(M^{m-1}L^{\frac{N+l-1}{N}})_{-}=0$ or $l=2i\ (i\in\mathbb{Z})$, we can obtain
\begin{equation*}
\sum\limits_{\alpha=1}^{n}\partial_{l,m,\alpha}^{*}L_{-}=(1-(-1)^{l})mN(M^{m-1}L^{\frac{N+l-1}{N}})_{-}=0.\ \ \ \ \ \ \ \ \ \ \ \
\end{equation*}
\end{proof}
\begin{proposition}
The additional symmetric flows $\sum\limits_{\alpha=1}^{n}\partial_{l,m,\alpha}^{*}$ which satisfy the condition $(M^{m-1}L^{\frac{N+l-1}{N}})_{-}=0$ or the condition $l=2i\ (i\in\mathbb{Z})$ commute with the flows $\partial_{2k+1,\beta}\ (2k+1\neq0\ mod\ N)$ of the B type multi-component Gelfand-Dickey hierarchy.
\end{proposition}
\begin{proof}
\begin{equation*}
\ \ \ \ \ \ \ \ \big[\sum\limits_{\alpha=1}^{n}\partial_{l,m,\alpha}^{*},\partial_{2k+1,\beta}\big]\phi=\Big(\sum\limits_{\alpha=1}^{n}\big[\partial_{l,m,\alpha}^{*},\partial_{2k+1,\beta}\big]\Big)\phi=
\sum\limits_{\alpha=1}^{n}\Big(\big[\partial_{l,m,\alpha}^{*},\partial_{2k+1,\beta}\big]\phi\Big),
\end{equation*}
\begin{equation*}
\begin{split}
&[\partial_{l,m,\alpha}^{*},\partial_{2k+1,\beta}]\phi\\
=&-\partial_{l,m,\alpha}^{*}((B_{2k+1,\beta})_{-}\phi)+\partial_{2k+1,\beta}((D_{l,m,\alpha})_{-}\phi)\\
=&[(D_{l,m,\alpha})_{-},(B_{2k+1,\beta})_{-}+\partial_{2k+1,\beta}]_{-}\phi+(B_{2k+1,\beta})_{-}(D_{l,m,\alpha})_{-}\phi\\ \ \ \ \ \ \ \
&+[\partial_{2k+1,\beta},(D_{l,m,\alpha})_{-}]_{-}\phi-(D_{l,m,\alpha})_{-}(B_{2k+1,\beta})_{-}\phi\\
=&[(D_{l,m,\alpha})_{-},(B_{2k+1,\beta})_{-}]\phi-[(D_{l,m,\alpha})_{-},(B_{2k+1,\beta})_{-}]_{-}\phi\\
=&0,
\end{split}
\end{equation*}
then $\big[\sum\limits_{\alpha=1}^{n}\partial_{l,m,\alpha}^{*},\partial_{2k+1,\beta}\big]\phi=0,$ so the proposition is proved.
\end{proof}
Then we consider a special additional flow $\big(l=-(2k-1),\ l-1=0\ mod\ N\big)$
\begin{equation*}
\sum\limits_{\alpha=1}^{n}\partial_{-(2k-1),1,\alpha}^{*}L^{\frac{2k}{N}}=-\sum\limits_{\alpha=1}^{n}\big[(D_{-(2k-1),1,\alpha})_{-},L^{\frac{2k}{N}}\big],\ 2k=0\ mod\ N,
\end{equation*}
after a series of calculations, we can get
\begin{equation*}
\sum\limits_{\alpha=1}^{n}\partial_{-(2k-1),1,\alpha}^{*}L^{\frac{2k}{N}}=\sum\limits_{\alpha=1}^{n}\big[(D_{-(2k-1),1,\alpha})_{+},L^{\frac{2k}{N}}\big]+4kI,\ \
\end{equation*}
When $2k=0\ mod\ N$, $L^{\frac{2k}{N}}$ is a differential operator, then
\begin{equation*}
\sum\limits_{\alpha=1}^{n}\big[L^{\frac{2k}{N}},\frac{1}{4k}(D_{-(2k-1),1,\alpha})_{+}\big]=I,\ \ \ \ \ \ \ \ \ \ \ \ \ \ \ \ \ \ \ \ \ \ \ \ \ \ \ \ \
\end{equation*}
thus $\Big[L^{\frac{2k}{N}},\frac{1}{4k}\big(\sum\limits_{\alpha=1}^{n}(D_{-(2k-1),1,\alpha})_{+}\big)\Big]=I\ $ is the string equation of the mcBGD hierarchy.

\subsection{$\tau$ function and Virasoro constraint of the mcBGD hierarchies}
\ \ \ \ The existence theorem of a $\tau$ function of the mcBGD hierarchy is given below.
\begin{theorem}\cite{20}
Suppose the $\tau$ function is a matrix T=($\tau_{\alpha,\beta}$) and $\tau_{\alpha,\alpha}$ is abbreviated as $\tau$, then there exists a function $\tau(\cdots,t_{2s-1,\gamma},\cdots)$, which makes
\begin{equation}\label{1}
\begin{cases}
\omega_{\alpha,\alpha}(t,z)=\frac{\tau(\cdots,t_{2s-1,\gamma}-\delta_{\gamma,\beta}\frac{1}{(2s-1)z^{2s-1}},\cdots)}{\tau(\cdots,t_{2s-1,\gamma},\cdots)},\\[5pt]
\omega_{\alpha,\beta}(t,z)=\frac{\tau_{\alpha,\beta}(\cdots,t_{2s-1,\gamma}-\delta_{\gamma,\beta}\frac{1}{(2s-1)z^{2s-1}},\cdots)}{z\cdot\tau(\cdots,t_{2s-1,\gamma},\cdots)},\ \alpha\neq\beta,\\
\ \ \ \ \ \ \ \ \ \ \ \ \ \ \ \ \ \ \ \ \ \ \ \ \ \ \ \ \ \ \ \ \ \ \ \ \ \ 2s-1\neq 0\ mod\ N,
\end{cases}
\end{equation}
only when $\gamma=\beta$, the time variable $t_{2s-1,\gamma}$ will be moved.\\
 Except for the above case $2s-1\neq 0\ mod\ N$, the other cases satisfy $\sum\limits_{\alpha=1}^{n}\partial_{k,\alpha}\tau=0,\ where\ \ k=0\ mod\ N\ or\ k=0\ mod\ 2.$
\end{theorem}
Next, starting from the additional symmetry of the wave operator $\omega$ and combining with the existence theorem of the $\tau$ function, we study the additional symmetry of the $\tau$ function.
\begin{proposition}\label{e}
$D_{-(2k-1),1,\alpha}\ (2k=0\ mod\ N,\ k>0)$ can be reduced to the following form
\begin{equation}
\begin{split}
&(D_{-(2k-1),1,\alpha})_{-}\\
=&2\phi(\sum\limits_{i=1}^{k}(2i-1)t_{2i-1,\alpha}A^{\frac{2i-2k-1}{N}}\partial^{2i-2k-1}E_{\alpha})\phi^{-1}\\
&+2\sum\limits_{i=k+1}^{\infty}(2i-1)t_{2i-1,\alpha}(L^{\frac{2i-2k-1}{N}}R_{\alpha})_{-}\ \ \ \ \ \ \ \ \ \ \ \ \ \ \ \ \ \ \ \ \\
&-2kL^{-\frac{2k}{N}}R_{\alpha},\ \ 2i-1\neq 0\ mod\ N.
\end{split}
\end{equation}
\end{proposition}
\begin{proof}
According to the additional symmetry of the mcBGD hierarchy, we will have
\begin{equation*}
\begin{split}
\ \ \ \ \ &(D_{-(2k-1),1,\alpha})_{-}\\
=&(ML^{-\frac{2k-1}{N}}R_{\alpha}-(-1)^{-(2k-1)}R_{\alpha}L^{\frac{l-1}{N}}M^{m}L^{\frac{1}{N}})_{-}\\
=&(ML^{-\frac{2k-1}{N}}R_{\alpha})_{-}+(R_{\alpha}L^{\frac{l-1}{N}}M^{m}L^{\frac{1}{N}})_{-}\\
=&2(\phi\Gamma A^{-\frac{2k-1}{N}}\partial^{-(2k-1)}E_{\alpha}\phi^{-1})_{-}-2k(\phi E_{\alpha} A^{-\frac{2k}{N}}\partial^{-2k}\phi^{-1})_{-}\\
=&2(\phi\sum\limits_{i=1}^{\infty}(2i-1)t_{2i-1,\alpha}A^{\frac{2i-2k-1}{N}}\partial^{2i-2k-1}E_{\alpha}\phi^{-1})_{-}-2kL^{-\frac{2k}{N}}R_{\alpha}\\
=&2\phi(\sum\limits_{i=1}^{k}(2i-1)t_{2i-1,\alpha}A^{\frac{2i-2k-1}{N}}\partial^{2i-2k-1}E_{\alpha})\phi^{-1}\\
&+2\sum\limits_{i=k+1}^{\infty}(2i-1)t_{2i-1,\alpha}(L^{-\frac{2i-2k-1}{N}}R_{\alpha})_{-}
-2kL^{-\frac{2k}{N}}R_{\alpha}.
\end{split}
\end{equation*}
\end{proof}
Based on the additional symmetry of the mcBGD hierarchy and the proposition \ref{e}, we can calculate
\begin{equation}\label{f}
\begin{split}
&\sum\limits_{\alpha=1}^{n}\partial_{-(2k-1),1,\alpha}^{*}\phi\\
=&-\sum\limits_{\alpha=1}^{n}(D_{-(2k-1),1,\alpha})_{-}\phi\\
=&\sum\limits_{\alpha=1}^{n}\big(-2\phi\sum\limits_{i=1}^{k}(2i-1)t_{2i-1,\alpha}A^{\frac{2i-2k-1}{N}}\partial^{2i-2k-1}E_{\alpha}\\
&+2\sum\limits_{i=k+1}^{\infty}(2i-1)t_{2i-1,\alpha}(\partial_{2(i-k)-1,\alpha}\phi)\big)
+2k\phi A^{-\frac{2k}{N}}\partial^{-2k}I,\\
&\ \ \ \ \ \ \ \ \ \ \ \ \ \ \ \ \ \ \ \ \ \ \ \ \ \ \ \ \ \ \ \ \ \ \ \ \ \ \ \ \ \ \ \ \  \ \ \ \ \ \ \ (2k=0\ mod\ N).
\end{split}
\end{equation}
Obviously, when both sides of equation \eqref{f} are simultaneously applied to a function $exp(A^{-\frac{1}{N}}xz)$, the equation is still valid. By further calculating with
\begin{equation*}
[\phi,t_{1,\alpha}]exp(A^{-\frac{1}{N}}xz)=(\partial_{z}\omega)exp(A^{-\frac{1}{N}}xz),\ \ \ \ \ \ \ \ \ \ \ \ \ \ \ \ \ \ \ \ \ \ \ \ \ \ \ \
\end{equation*}
\begin{equation*}
\phi\partial^{-l}exp(A^{-\frac{1}{N}}xz)=\phi z^{-l}exp(A^{-\frac{1}{N}}xz)=z^{-l}\omega(exp(A^{-\frac{1}{N}}xz)),\ \ \
\end{equation*}
we can get
\begin{equation*}
\begin{split}
&(\sum\limits_{\alpha=1}^{n}\partial_{-(2k-1),1,\alpha}^{*}\omega)exp(A^{-\frac{1}{N}}xz)\\
=&(-\sum\limits_{\alpha=1}^{n}(D_{-(2k-1),1,\alpha})_{-}\phi)exp(A^{-\frac{1}{N}}xz)\\
=&-2z^{-2k+1}(\partial_{z}\omega)exp(A^{-\frac{1}{N}}xz)+2kz^{-2k}\omega(exp(A^{-\frac{1}{N}}xz))\\
&-2(\sum\limits_{i=1}^{k}\sum\limits_{\alpha=1}^{n}(2i-1)t_{2i-1,\alpha}A^{\frac{2i-2k-1}{N}}\partial^{2i-2k-1}E_{\alpha}\omega)exp(A^{-\frac{1}{N}}xz)\\
&+2\sum\limits_{i=k+1}^{\infty}\sum\limits_{\alpha=1}^{n}(2i-1)t_{2i-1,\alpha}(\partial_{2(i-k)-1,\alpha}\omega)exp(A^{-\frac{1}{N}}xz),
\end{split}
\end{equation*}
therefore
\begin{equation}
\begin{split}
&\sum\limits_{\alpha=1}^{n}\partial_{-(2k-1),1,\alpha}^{*}\omega\\
=&-2z^{-2k+1}(\partial_{z}\omega)I+2kz^{-2k}\omega I\\
&-2(\sum\limits_{i=1}^{k}\sum\limits_{\alpha=1}^{n}(2i-1)t_{2i-1,\alpha}A^{\frac{2i-2k-1}{N}}\partial^{2i-2k-1}E_{\alpha}\omega)\ \ \ \ \ \ \ \ \ \ \
\ \ \ \ \\
&+2\sum\limits_{i=k+1}^{\infty}\sum\limits_{\alpha=1}^{n}(2i-1)t_{2i-1,\alpha}(\partial_{2(i-k)-1,\alpha}\omega),\\
\end{split}
\end{equation}
correspondingly, we can obtain
\begin{equation}
\begin{split}
&\sum\limits_{\alpha=1}^{n}\partial_{-(2k-1),1,\alpha}^{*}\omega_{\gamma,\beta}\\
=&\sum\limits_{\alpha=1}^{n}\big(-2z^{-2k+1}\delta_{\alpha,\beta}(\partial_{z}\omega_{\gamma,\beta})+2kz^{-2k}\delta_{\alpha,\beta}\omega_{\gamma,\beta}\\
&-2(\sum\limits_{i=1}^{k}(2i-1)t_{2i-1,\alpha}A^{\frac{2i-2k-1}{N}}\partial^{2i-2k-1}\delta_{\alpha,\beta}\omega_{\gamma,\beta})\ \ \ \ \ \ \ \ \ \ \
\ \\
&+2\sum\limits_{i=k+1}^{\infty}(2i-1)t_{2i-1,\alpha}(\partial_{2(i-k)-1,\alpha}\omega_{\gamma,\beta})\big).\\
\end{split}
\end{equation}
We can find that $\sum\limits_{\alpha=1}^{n}\partial_{-(2k-1),1,\alpha}^{*}\phi=0$ and $\sum\limits_{\alpha=1}^{n}\partial_{-(2k-1),1,\alpha}^{*}\omega=0\ $ are equivalent. Thus, we can get the constraints of the tau function by equation \eqref{1}.\\
For $\sum\limits_{\alpha=1}^{n}\partial_{-(2k-1),1,\alpha}^{*}\omega\ (2k=0\ mod\ N,\ k>0)$, it is difficult to find a unified simplified form, so in this paper we only discuss the case where $k$ is a positive integer.
\begin{proposition}
The action of additional symmetries $\sum\limits_{\alpha=1}^{n}\partial_{-(2k-1),1,\alpha}^{*}\ (2k=0\ mod\ N,\ k>0)$ on $(\tau_{\gamma,\beta})$ is given by
\begin{equation}
\begin{split}
&\sum\limits_{\alpha=1}^{n}\partial_{-(2k-1),1,\alpha}^{*}\tau_{\gamma,\beta}\\
=&\sum\limits_{\alpha=1}^{n}\big((\frac{1}{2}\sum\limits_{m=2k}^{\infty}(2m-1)t_{2m-1,\alpha}\partial_{2(m-k)-1,\alpha}\\
&+\frac{1}{8}\sum\limits_{m+s=2k}(2m-1)(2s-1)t_{2m-1,\alpha}t_{2s-1,\alpha})\big)\tau_{\gamma,\beta}+(\sum\limits_{\alpha=1}^{n}c_{\alpha,\beta})\tau_{\gamma,\beta}.
\end{split}
\end{equation}
where $t_{2m-1,\alpha},\ t_{2s-1,\alpha}$ are arguments of the element $\tau_{\gamma,\beta}$.
\end{proposition}
Taking
\begin{equation}
\begin{split}
L_{-(2k-1)}=&\sum\limits_{\alpha=1}^{n}\big(\frac{1}{2}\sum\limits_{m=2k}^{\infty}(2m-1)t_{2m-1,\alpha}\partial_{2(m-k)-1,\alpha}\ \ \ \ \\
&+\frac{1}{8}\sum\limits_{m+s=2k}(2m-1)(2s-1)t_{2m-1,\alpha}t_{2s-1,\alpha}+c_{\alpha,\beta}\big),
\end{split}
\end{equation}
where $2k=0\ mod\ N,\ k>0$.\\
We can verify that $L_{-(2k-1)}$ satisfies the Virasoro exchange relation
\begin{equation}
[L_{-m},L_{-n}]=(-m+n)L_{-(m+n)},\ \ m,n=1,2,\cdots.\ \ \ \ \ \
\end{equation}
so the Virasoro constraint of the $\tau$ function of the mcBGD hierarchy is
\begin{equation}
L_{-(2k-1)}\tau=0,\ \ 2k=0\ mod\ N,\ k>0.\ \ \ \ \ \ \ \ \ \ \ \ \ \ \ \ \ \ \ \ \ \ \ \ \ \ \ \ \ \ \ \ \ \
\end{equation}

\section{C type multi-component Gelfand-Dickey hierarchies}
\ \ \ \ The Lax operators of the mcCGD hierarchy are composed of the operators $L$ and $R_{\alpha}$ of the mcGD hierarchy and the constraint conditions $L^{*}=(-1)^{N}L,\ R_{\alpha}^{*}=R_{\alpha}$. It is clear that the mcCGD hierarchy can be divided into two types according to the parity of $N$. The odd mcCGD hierarchy can be constructed according to the differential operator
\begin{equation*}
\ \ \ \ \ \ \ \ L=A\partial^{N}+u_{2}\partial^{N-2}+u_{3}\partial^{N-3}+\cdots +u_{0},\ N\neq 0\ mod\ 2,
\end{equation*}
\begin{equation*}
R_{\alpha}=\sum\limits_{i=1}^{\infty}R_{i\alpha}\partial^{-i},\ \ \alpha=1,2,\cdots ,n,\ \ \ \ \ \ \ \ \ \ \ \ \ \ \ \ \ \ \ \
\end{equation*}
which satisfy $L*=-L$ and $R_{\alpha}^{*}=R_{\alpha}$.\\
The even mcCGD hierarchy can also be constructed according to the differential operator
\begin{equation*}
\ \ \ \ \ \ \ \ \ L=A\partial^{N}+u_{2}\partial^{N-2}+u_{3}\partial^{N-3}+\cdots +u_{0},\ N=0\ mod\ 2,
\end{equation*}
\begin{equation*}
R_{\alpha}=\sum\limits_{i=1}^{\infty}R_{i\alpha}\partial^{-i},\ \ \alpha=1,2,\cdots ,n,\ \ \ \ \ \ \ \ \ \ \ \ \ \ \ \ \ \ \ \
\end{equation*}
which satisfy $L*=L$ and $R_{\alpha}^{*}=R_{\alpha}$.\\
Though their constraints are different, their additional symmetric generators, surviving flows, string equation and so on are the same.

Similar to the Lax equations of the mcBGD hierarchy, the Lax equations of the mcCGD hierarchy does not have even flows, that is, $u_{i}=u_{i}(x;t_{1},t_{3},\cdots)$.
\begin{proposition}
The C type multi-component Gelfand-Dickey hierarchy has only odd flows.
\end{proposition}
\begin{proof}
From the foregoing, we already know the Lax equations of the mcCGD hierarchy. Combined with the constraints of its Lax equations, we can get
\begin{equation*}
(L^{\frac{n}{N}})^{*}=(-1)^{n}L^{\frac{n}{N}},\ \ \ \ \
\end{equation*}
\begin{equation*}
\ \ \ (L^{\frac{n}{N}}R_{\alpha})^{*}=(-1)^{n}R_{\alpha}L^{\frac{n}{N}}.
\end{equation*}
Then we simplify $\partial_{n,\alpha}L^{*}$ and $\partial_{n,\alpha}R_{\alpha}^{*}$ in two ways
\begin{equation}\label{a1}
\begin{split}
\partial_{n,\alpha}L^{*}=(\partial_{n,\alpha}L)^{*}=[B_{n,\alpha},L]^{*}=-[B_{n,\alpha},L]=[B_{n,\alpha},L^{*}],\ \ \ \ \ \ \ \ \ \ \ \ \ \ \ \ \ \ \ \ \ \ \ \ \ \ \ \ \ \ \ \\
\partial_{n,\alpha}L^{*}=[B_{n,\alpha},L]^{*}=(B_{n,\alpha}L-L B_{n,\alpha})^{*}=L^{*}B_{n,\alpha}^{*}-B_{n,\alpha}^{*}L^{*}=[-B_{n,\alpha}^{*},L^{*}],\ \ \ \ \ \ \ \\[-5pt]
\end{split}
\end{equation}
\begin{equation}\label{b1}
\begin{split}
\partial_{n,\alpha}R_{\beta}^{*}=(\partial_{n,\alpha}R_{\beta})^{*}=[B_{n,\alpha},R_{\beta}]^{*}=[B_{n,\alpha},R_{\beta}]=[B_{n,\alpha},R_{\beta}^{*}],\ \ \ \ \ \ \ \ \ \ \ \ \ \ \ \ \ \ \ \ \ \ \ \ \ \ \ \ \ \ \ \ \ \\
\partial_{n,\alpha}R_{\beta}^{*}=[B_{n,\alpha},R_{\beta}]^{*}=(B_{n,\alpha}R_{\beta}-R_{\beta} B_{n,\alpha})^{*}=R_{\beta}^{*}B_{n,\alpha}^{*}-B_{n,\alpha}^{*}R_{\beta}^{*}=[-B_{n,\alpha}^{*},R_{\beta}^{*}],\ \ \ \ \ \
\end{split}
\end{equation}
comparing \eqref{a1} and \eqref{b1}, we can find
\begin{equation}\label{c1}
\ B_{n,\alpha}^{*}=-B_{n,\alpha},\ \ n\neq 0\ mod\ N.
\end{equation}
Then we combine the definition of $B_{n,\alpha}\ (n\neq 0\ mod\ N)$ to solve its adjoint
\begin{equation}\label{d1}
B_{n,\alpha}^{*}=((L^{\frac{n}{N}}R_{\alpha})_{+})^{*}=(-1)^{n}(R_{\alpha}L^{\frac{n}{N}})_{+}=(-1)^{n}B_{n,\alpha},\ \ n\neq 0\ mod\ N.\ \
\end{equation}
By analyzing equation \eqref{c1} and \eqref{d1}, we find that $n$ can only take odd numbers.
\end{proof}
So its Lax equations can be defined as
\begin{equation*}
\ \ \ \  \ \ \ \ \ \partial_{2n+1,\alpha}L=[B_{2n+1,\alpha},L],\\ \ \ \ \ \partial_{2n+1,\alpha}R_{\beta}=[B_{2n+1,\alpha},R_{\beta}],\ \ 2n+1\neq 0\ mod\ N,\ \
\end{equation*}
\begin{equation}
\ \ \ \  \ \ \ \ \ \  \ \ \ \sum\limits_{\alpha=1}^{n}\partial_{(2j^{'}+1)N,\alpha}L=0,\ \ \ \ \ \ \ \
\sum\limits_{\alpha=1}^{n}\partial_{(2j^{'}+1)N,\alpha}R_{\beta}=0,\ \ j^{'}=0,1,2,\cdots,\ \ \ \ \
\end{equation}
where $\partial_{2n+1,\alpha}=\frac{\partial}{\partial t_{2n+1,\alpha}}$ and $B_{2n+1,\alpha}=(L^{\frac{2n+1}{N}}R_{\alpha})_{+}$.\\
Taking a dressing operator $\phi=\phi(A\partial^{N})=\sum\limits_{i=1}^{\infty}\alpha_{i}(A\partial^{N})^{-i}\ (\alpha_{0}=I,\ \phi^{*}=\phi^{-1})$, then the wave function of the mcCGD hierarchy is $W(t,z)=\phi e^{\xi(t,z)}=\omega(t,z)e^{\xi(t,z)}$, where
\begin{equation*}
\xi(t,z)=\sum\limits_{i=1}^{\infty}\sum\limits_{\alpha=1}^{n}t_{2i-1,\alpha}E_{\alpha}z^{2i-1}.\ \ \ \  \ \ \ \ \ \ \ \ \ \ \ \ \ \ \ \ \ \ \ \ \ \ \ \ \ \ \ \  \ \ \ \ \ \ \ \ \ \ \ \ \ \ \ \
\end{equation*}
In addition, the adjoint wave function of the mcCGD hierarchy is
\begin{equation*}
W^{*}(t,z)=(\phi^{*})^{-1}e^{-\xi(t,z)}=\phi e^{-\xi(t,z)}.\ \ \ \  \ \ \ \ \ \ \ \ \ \ \ \ \ \ \ \ \ \ \ \  \ \ \ \ \ \ \ \ \ \ \ \ \ \ \ \ \
\end{equation*}
Then we can get the Sato equations $\partial_{2n+1,\alpha}\phi=-(B_{2n+1,\alpha})_{-}\phi\ (2n+1\neq 0\ mod\ N)$ and $\sum\limits_{\alpha=1}^{n}\partial_{(2j^{'}+1)N,\alpha}\phi=0\ (j^{'}=0,1,2,\cdots)$ of the mcCGD hierarchy.\\
And the Lax equations of the mcCGD hierarchy can also be obtained by the compatibility conditions of the following linear partial differential equations
\begin{equation}
\begin{split}
L^{\frac{2n+1}{N}}W=z^{2n+1}W,\ \ \ \ \ \partial_{2n+1,\alpha}W=B_{2n+1,\alpha}W,\ 2n+1\neq 0\ mod\ N,\\
\sum\limits_{\alpha=1}^{n}\partial_{(2j^{'}+1)N,\alpha}W=L^{(2j^{'}+1)}W,\ \ j^{'}=0,1,2,\cdots, \ \ \ \ \ \ \ \ \ \ \ \ \ \ \ \ \ \ \ \ \ \ \ \ \ \
\end{split}
\end{equation}
where $W(t,z)$ is the wave function of the mcCGD hierarchy.

\subsection{Additional symmetry of the mcCGD hierarchies}
\ \ \ \ From the above knowledge points related to the mcCGD hierarchy, we can find that many of its definitions and properties are similar to those of the mcBGD hierarchy, and the biggest difference lies in the related definitions and operations of their adjoint. Next, we mainly discuss the knowledge of adjoint in detail. Similarly, we first give the Orlov-Shulman's operator $M=\phi\Gamma\phi^{-1}$ of the mcCGD hierarchy, where
\begin{equation*}
\Gamma=\sum\limits_{i=1}^{\infty}\sum\limits_{\alpha=1}^{n}(2i-1)t_{2i-1,\alpha}A^{\frac{2i-2}{N}}\partial^{2i-2}E_{\alpha},\ \ 2i-1\neq 0\ mod\ N,\ \ \ \ \ \ \
\end{equation*}
then we can get
\begin{equation*}
M^{*}=(\phi\Gamma\phi^{-1})^{*}=\phi\Gamma\phi^{-1}=M. \ \ \ \ \ \ \ \ \ \  \ \ \ \ \ \ \ \ \ \ \ \ \ \ \ \ \ \ \ \ \ \ \ \ \ \ \ \ \ \ \ \ \
\end{equation*}
After some necessary knowledge reserve, we will give the definition of the additional symmetry of the mcCGD hierarchy. The $C_{l,m,\alpha}$ in the additional symmetry of the mcCGD hierarchy is obviously different from the $D_{l,m,\alpha}$ in the definition of the additional symmetry of the mcBGD hierarchy. The fundamental reason is that their Lax operators have different constraints.
\begin{proposition}
In the C type multi-component Gelfand-Dickey hierarchy, the additional symmetric generators $C_{l,m,\alpha}$ should satisfy $(C_{l,m,\alpha})^{*}=-C_{l,m,\alpha}$, then $C_{l,m,\alpha}$ can be taken as
\begin{equation}
C_{l,m,\alpha}=M^{m}L^{\frac{l}{N}}R_{\alpha}-(-1)^{l}R_{\alpha}L^{\frac{l}{N}}M^{m}.\ \ \ \ \ \ \ \ \ \ \ \ \ \ \ \ \ \ \ \ \ \ \ \
\end{equation}
\end{proposition}
\begin{proof}
According to the the additional symmetry of the mcCGD hierarchy, we can know that
\begin{equation*}
\partial_{l,m,\alpha}^{*}\phi^{*}=(\partial_{l,m,\alpha}^{*}\phi)^{*}=(-(C_{l,m,\alpha})_{-}\phi)^{*}=-\phi^{*}(C_{l,m,\alpha})_{-}^{*},\ \ \
\end{equation*}
and according to $\phi^{*}=\phi^{-1}$, we can deduce that
\begin{equation*}
\ \ \ \ \ \ \ \ \ \ \ \ \ \ \ \ \ \ \partial_{l,m,\alpha}^{*}\phi^{*}=\partial_{l,m,\alpha}^{*}\phi^{-1}=-\phi^{-1}(\partial_{l,m,\alpha}^{*}\phi)\phi^{-1}=\phi^{-1}(C_{l,m,\alpha})_{-}=\phi^{*}(C_{l,m,\alpha})_{-},
\end{equation*}
by comparing the results of the above two different methods, we can find that
\begin{equation*}
(C_{l,m,\alpha})_{-}=-(C_{l,m,\alpha})_{-}^{*},
\end{equation*}
so $\ C_{l,m,\alpha}=-(C_{l,m,\alpha})^{*}.$\\
Combined with the additional symmetric generators$\ M^{m}L^{\frac{l}{N}}R_{\alpha}\ $of the mcGD hierarchy, we will have
\begin{equation*}
(M^{m}L^{\frac{l}{N}}R_{\alpha})^{*}=R_{\alpha}^{*}(L^{\frac{l}{N}})^{*}(M^{m})^{*}=(-1)^{l}R_{\alpha}L^{\frac{l}{N}}M^{m}, \ \ \ \ \ \ \
\end{equation*}
then it's easy to construct
\begin{equation*}
C_{l,m,\alpha}=M^{m}L^{\frac{l}{N}}R_{\alpha}-(-1)^{l}R_{\alpha}L^{\frac{l}{N}}M^{m}. \ \ \ \ \ \ \ \ \ \ \ \ \ \ \ \ \ \ \ \ \ \ \ \
\end{equation*}
\end{proof}
\begin{definition}
The solution of the differential equation
\begin{equation}
\frac{\partial\phi}{\partial_{l,m,\alpha}^{*}}=-(C_{l,m,\alpha})_{-}\phi,
\end{equation}
where
\begin{equation}
\partial_{l,m,\alpha}^{*}=\frac{\partial}{\partial t_{l,m,\alpha}^{*}},\ \ C_{l,m,\alpha}=M^{m}L^{\frac{l}{N}}R_{\alpha}-(-1)^{l}R_{\alpha}L^{\frac{l}{N}}M^{m},\ \ \ \ \ \ \
\end{equation}
is called the additional symmetry of the C type multi-component Gelfand-Dickey hierarchy.
\end{definition}
Similarly, for the differential operator $C_{l,m,\alpha}$, assuming its negative part disappears and let the operator $(C_{l,m,\alpha})_{+}$ act on $W$, we can get an equation related to $z$
\begin{equation}
(C_{l,m,\alpha})_{+}W=(1-(-1)^{l})z^{l}E_{\alpha}\partial^{m}_{z}W-(-1)^{l}mlz^{l-1}E_{\alpha}\partial^{m-1}_{z}W.
\end{equation}
Note:This system can also be rewritten into a linear equation for the isomonodromy problem.\\
The following is a brief description of the construction method of $C_{l,m,\alpha}$.
According to the additional symmetry of the mcCGD hierarchy, we can deduce
\begin{equation}
\ \ \ \ \ \partial_{l,m,\alpha}^{*}L=-[(C_{l,m,\alpha})_{-},L], \ \ \ \ \ \ \
\partial_{l,m,\alpha}^{*}R_{\beta}=-[(C_{l,m,\alpha})_{-},R_{\beta}].
\end{equation}
For the additional flows of the mcCGD hierarchy, they have a fraction of the flows that can survive.
\begin{theorem}
In the additional flows of the C type multi-component Gelfand-Dickey hierarchy, the flows which satisfy the condition $(M^{m-1}L^{\frac{N+l-1}{N}})_{-}=0$ and are shaped like $\sum\limits_{\alpha=1}^{n}\partial_{l,m,\alpha}^{*}$ or the condition $l=2i\ (i\in\mathbb{Z})$ and are shaped like $\sum\limits_{\alpha=1}^{n}\partial_{l,m,\alpha}^{*}$ can survive.
\end{theorem}
\begin{proof}
For the Lax operator $L$ of the mcCGD hierarchy, its negative part is also equal to zero.\\
And we know $\partial_{l,m,\alpha}^{*}L=-[(C_{l,m,\alpha})_{-},L]$, so let's consider
\begin{equation*}
\begin{split}
\ \ \ \ \ \ \ \ \ \ \partial_{l,m,\alpha}^{*}L_{-}=&-[(C_{l,m,\alpha})_{-},L]_{-}=-(\phi[\Gamma^{m}A^{\frac{l}{N}}\partial^{l}E_{\alpha},A\partial^{N}]\phi^{-1})_{-}\\
&+(-1)^{l}(\phi[E_{\alpha}A^{\frac{l}{N}}\partial^{l}\Gamma^{m},A\partial^{N}]\phi^{-1})_{-},
\end{split}
\end{equation*}
after some deductions, the upper formula can be reduced to
\begin{equation*}
\partial_{l,m,\alpha}^{*}L_{-}=\partial_{l,m,\alpha}^{*}L_{-}=(1-(-1)^{l})mN(M^{m-1}L^{\frac{N+l-1}{N}}R_{\alpha})_{-},\ \
\end{equation*}
only when $(M^{m-1}L^{\frac{N+l-1}{N}})_{-}=0$ or $l=2i\ (i\in\mathbb{Z})$, we can obtain
\begin{equation*}
\sum\limits_{\alpha=1}^{n}\partial_{l,1,\alpha}^{*}L_{-}=(1-(-1)^{l})mN(M^{m-1}L^{\frac{N+l-1}{N}})_{-}=0.\ \ \ \ \ \ \ \
\end{equation*}
\end{proof}
The surviving additional flows of the mcCGD hierarchy satisfy the following propositions.
\begin{proposition}
The additional symmetric flows $\sum\limits_{\alpha=1}^{n}\partial_{l,m,\alpha}^{*}$ which satisfy the condition $(M^{m-1}L^{\frac{N+l-1}{N}})_{-}=0$ or the condition $l=2i\ (i\in\mathbb{Z})$ commute with the flows $\partial_{2k+1,\beta}\ (2k+1\neq 0\ mod\ N)$ of the C type multi-component Gelfand-Dickey hierarchy.
\end{proposition}
\begin{proof}
\begin{equation*}
\ \ \ \ \ \ \ \ \big[\sum\limits_{\alpha=1}^{n}\partial_{l,m,\alpha}^{*},\partial_{2k+1,\beta}\big]\phi=\Big(\sum\limits_{\alpha=1}^{n}\big[\partial_{l,m,\alpha}^{*},\partial_{2k+1,\beta}\big]\Big)\phi=
\sum\limits_{\alpha=1}^{n}\Big(\big[\partial_{l,m,\alpha}^{*},\partial_{2k+1,\beta}\big]\phi\Big),
\end{equation*}
\begin{equation*}
\begin{split}
&[\partial_{l,m,\alpha}^{*},\partial_{2k+1,\beta}]\phi\\
=&-\partial_{l,m,\alpha}^{*}((B_{2k+1,\beta})_{-}\phi)+\partial_{2k+1,\beta}((C_{l,m,\alpha})_{-}\phi)\\
=&[(C_{l,m,\alpha})_{-},(B_{2k+1,\beta})_{-}+\partial_{2k+1,\beta}]_{-}\phi+(B_{2k+1,\beta})_{-}(C_{l,m,\alpha})_{-}\phi\\ \ \ \ \ \ \ \
&+[\partial_{2k+1,\beta},(C_{l,m,\alpha})_{-}]_{-}\phi-(C_{l,m,\alpha})_{-}(B_{2k+1,\beta})_{-}\phi\\
=&[(C_{l,m,\alpha})_{-},(B_{2k+1,\beta})_{-}]\phi-[(C_{l,m,\alpha})_{-},(B_{2k+1,\beta})_{-}]_{-}\phi\\
=&0,
\end{split}
\end{equation*}
then $\big[\sum\limits_{\alpha=1}^{n}\partial_{l,m,\alpha}^{*},\partial_{2k+1,\beta}\big]\phi=0,$ so the proposition is proved.
\end{proof}
After some calculations, we find that whether the odd mcCGD hierarchy or the even mcCGD hierarchy, their string equations are the same. The derivation process is given below. First we consider a additional flow
\begin{equation*}
\sum\limits_{\alpha=1}^{n}\partial_{-(2k-1),1,\alpha}^{*}L^{\frac{2k}{N}}=-\sum\limits_{\alpha=1}^{n}\big[(C_{-(2k-1),1,\alpha})_{-},L^{\frac{2k}{N}}\big],\ 2k=0\ mod\ N,\ \ \ \ \ \ \ \ \ \ \ \ \ \
\end{equation*}
after a series of calculations, we will obtain
\begin{equation*}
\sum\limits_{\alpha=1}^{n}\partial_{-(2k-1),1,\alpha}^{*}L^{\frac{2k}{N}}=\sum\limits_{\alpha=1}^{n}\big[(C_{-(2k-1),1,\alpha})_{+},L^{\frac{2k}{N}}\big]+4kI,\ \ \ \ \ \
\end{equation*}
When $2k=0\ mod\ N$, $L^{\frac{2k}{N}}$ is a differential operator, then
\begin{equation*}
\sum\limits_{\alpha=1}^{n}\big[L^{\frac{2k}{N}},\frac{1}{4k}(C_{-(2k-1),1,\alpha})_{+}\big]=I,\ \ \ \ \ \ \ \ \ \ \ \ \ \ \ \ \ \ \ \ \ \ \ \ \ \ \ \ \
\end{equation*}
thus $\Big[L^{\frac{2k}{N}},\frac{1}{4k}\big(\sum\limits_{\alpha=1}^{n}(C_{-(2k-1),1,\alpha})_{+}\big)\Big]=I\ $ is the string equation of the mcCGD hierarchy.
In some literatures, the string equation is described as the form of $[P,Q]=1$, where $P$ and $Q$ are differential operators, which is equivalent to the string equation derived by us, and refers to the condition that the operator is independent of variables.
\\[5pt]

{\bf {Acknowledgements:}}
Chuanzhong Li  is  supported by the National Natural Science Foundation of China under Grant No. 11571192 and K. C. Wong Magna Fund in
Ningbo University.\\

\ \ \ \
\end{document}